\documentclass[letterpaper, 10 pt,journal]{ieeetran}

\IEEEoverridecommandlockouts                              % This command is only needed if 
\usepackage{graphics} % for pdf, bitmapped graphics files
\usepackage{graphicx}

\usepackage{times} % assumes new font selection scheme installed
\usepackage{amsmath} % assumes amsmath package installed
\usepackage{amssymb}  % assumes amsmath package installed

\usepackage{tikz} % <--- if tikz pictures are used
\usepackage{pgfplots} % <--- if pgfplots is used

\usepackage{algpseudocode}%
\usepackage{algorithm}
\usepackage{epstopdf}

\usepackage{color}

\usepackage{lipsum}% http://ctan.org/pkg/lipsum
\usepackage{verbatim}

\usepackage{amsthm}

\usepackage{mathtools}

\newtheorem{theorem}{\bf Theorem} \newtheorem{definition}[theorem]{\bf Definition} 
\newtheorem{lemma}[theorem]{\bf Lemma} \newtheorem{remark}[theorem]{\bf Remark}
  \newtheorem{proposition}[theorem]{\bf Proposition} 
\newtheorem{assumption}[theorem]{\bf Assumption}  
\newtheorem{Algorithm}[theorem]{\bf Algorithm}

\title{\LARGE \bf
Robust Constraint Satisfaction in Data-Driven MPC
}

\author{Julian Berberich$^1$, Johannes K\"ohler$^1$, Matthias A. M\"uller$^2$, Frank Allg\"ower$^1$% <-this % stops a space
\thanks{
This work was funded by Deutsche Forschungsgemeinschaft (DFG, German Research Foundation) under Germany's Excellence Strategy - EXC 2075 - 390740016.
The authors thank the International Max Planck Research School for Intelligent Systems (IMPRS-IS) for supporting Julian Berberich, and the International Research Training Group “Soft Tissue Robotics” (GRK 2198/1).}
\thanks{$^1$Julian Berberich, Johannes K\"ohler, and Frank Allg\"ower are with the Institute for Systems Theory and Automatic Control, University of Stuttgart, 70550 Stuttgart, Germany (email:$\{$ julian.berberich, johannes.koehler, frank.allgower$\}$@ist.uni-stuttgart.de)}
\thanks{$^2$Matthias A. M\"uller is with the Leibniz University Hannover, Institute of Automatic Control, 30167 Hannover, Germany (e-mail:mueller@irt.uni-hannover.de)}}

\begin{document}
\IEEEpubid{\begin{minipage}{\textwidth}\ \\[12pt] \\ \\
\copyright 2020 IEEE. This version has been accepted for publication in Proc. IEEE Conference on Decision and Control (CDC), 2020. Personal use of this material is permitted. Permission from IEEE must be obtained for all other uses, in any current or future media, including reprinting/republishing this material for advertising or promotional purposes, creating new collective works, for resale or redistribution to servers or lists, or reuse of any copyrighted component of this work in other works.
\end{minipage}}

\maketitle
%\thispagestyle{empty}
%\pagestyle{empty}

%%%%%%%%%%%%%%%%%%%%%%%%%%%%%%%%%%%%%%%%%%%%%%%%%%%%%%%%%%%%%%%%%%%%%%%%%%%%%%%%
\begin{abstract}
We propose a purely data-driven model predictive control (MPC) scheme to control unknown linear time-invariant systems with guarantees on stability and constraint satisfaction in the presence of noisy data.
The scheme predicts future trajectories based on data-dependent Hankel matrices, which span the full system behavior if the input is persistently exciting.
This paper extends previous work on data-driven MPC by including a suitable constraint tightening which ensures that the closed-loop trajectory satisfies desired pointwise-in-time output constraints.
Furthermore, we provide estimation procedures to compute system constants related to controllability and observability, which are required to implement the constraint tightening.
The practicality of the proposed approach is illustrated via a numerical example.
\end{abstract}
%
%!TEX root = ./DataDriven_IO.tex
%%%%%%%%%%%%%%%%%%%%%%%%%%%%%%%%%%%%%%%%%%%%%%%%%%%%%%%%%%%%%%%%%%%%%%%%%%%%%%%%

\section{Introduction}
Satisfying hard input and output constraints is a challenging control problem, in particular in the presence of uncertainty.
Perhaps the most suitable method to deal with constraints is model predictive control (MPC), which relies on the repeated solution of an open-loop optimal control problem~\cite{Rawlings17}.
While most of the existing literature on MPC is based on some form of model knowledge, there have been various recent contributions in the field of \emph{data-driven} MPC~\cite{Yang15,Coulson19,Coulson19b,Berberich19b,Berberich20c}.
In these works, the key idea is to replace the standard model commonly used for prediction in MPC by Hankel matrices consisting of one input-output data trajectory.
It is shown in~\cite{Willems05} for linear time-invariant (LTI) systems that, if the input generating the data is persistently exciting, then these Hankel matrices span the full system behavior and hence, they provide an implicit model characterizing all possible future trajectories.
A first theoretical analysis of closed-loop properties under data-driven MPC based on~\cite{Willems05} is provided in~\cite{Berberich19b}, both for the nominal case and for the scenario that the output measurements are perturbed by noise.
In particular, closed-loop guarantees on (practical) stability are derived and connections between the richness of the data, design parameters and desirable closed-loop properties are revealed.
This formulation has been extended to setpoint tracking in~\cite{Berberich20c}.
While~\cite{Berberich19b} and~\cite{Berberich20c} both prove closed-loop constraint satisfaction of the output for the nominal noise-free case, deriving an analogous result \emph{in the presence of noise} remains an open problem.
In this paper, we augment the MPC scheme proposed in~\cite{Berberich19b} by an output constraint tightening, which ensures that the true closed-loop output satisfies pointwise-in-time constraints.
The key contributions of this paper are the proposed constraint tightening, the proof of closed-loop constraint satisfaction, as well as a data-based estimation procedure for two system constants that are required to implement the constraint tightening.
\IEEEpubidadjcol

Generally speaking, data-driven MPC schemes~\cite{Yang15,Coulson19,Coulson19b,Berberich19b,Berberich20c} are inherently output-feedback MPC schemes, since no state measurement is used. 
Due to the presence of measurement noise in the data, the proposed approach can be seen as an alternative to model-based robust output-feedback MPC~\cite{Chisci02,Mayne06}, MPC for parametric uncertainty (e.g., adaptive MPC~\cite{Lorenzen19}), or learning-based MPC~\cite{Aswani13,Hewing20}.
Possibly competitive output-feedback MPC schemes for uncertain systems have recently been proposed in~\cite{tanaskovic2014adaptive,Terzi19}, using set membership methods to obtain a model.
We compare the proposed data-driven MPC scheme with the scheme presented in~\cite{Terzi19} via a numerical example.

The remainder of the paper is structured as follows.
After introducing some preliminaries in Section~\ref{sec:setting}, we present the proposed MPC scheme in Section~\ref{sec:MPC} together with the novel constraint tightening.
Section~\ref{sec:theory} contains a theoretical analysis of this MPC scheme, including a proof of recursive feasibility and constraint satisfaction of the closed loop.
In Section~\ref{sec:system_constants}, we propose a data-based estimation procedure for the system constants used in the constraint tightening.
Finally, Sections~\ref{sec:example} and~\ref{sec:conclusion} contain a numerical example illustrating the applicability of the presented approach and a conclusion, respectively.

\section{Setting}\label{sec:setting}
For a vector $x\in\mathbb{R}^n$, we denote by $\lVert x\rVert_1,\lVert x\rVert_2$, and $\lVert x\rVert_\infty$ the $1$-, $2$-, and $\infty$-norm, respectively, and similarly for induced matrix norms.
Further, we define $\lVert x\rVert_P=\sqrt{x^\top Px}$ for a positive definite matrix $P=P^\top\succ0$.
The set of integers in the interval $[a,b]$ is denoted by $\mathbb{I}_{[a,b]}$.
We denote the interior of a set $M$ by $\text{int}(M)$.
For a sequence $\{u_k\}_{k=0}^{N-1}$, we define the Hankel matrix
\begin{align*}
H_L(u)=\begin{bmatrix}u_0&u_1&\dots&u_{N-L}\\
u_1&u_2&\dots&u_{N-L+1}\\
\vdots&\vdots&\ddots&\vdots\\
u_{L-1}&u_L&\dots&u_{N-1}
\end{bmatrix}.
\end{align*}
Further, for integers $a,b$, we abbreviate a stacked window of the sequence as $u_{[a,b]}=\begin{bmatrix}u_a^\top&\dots&u_b^\top\end{bmatrix}^\top$.
Moreover, $u$ will denote either the sequence itself or the stacked vector $u_{[0,N-1]}$ containing all of its components.
We consider the following standard definition of persistence of excitation.
\begin{definition}\label{def:pe}
We say that a sequence $\{u_k\}_{k=0}^{N-1}$ with $u_k\in\mathbb{R}^m$ is persistently exciting of order $L$ if $\text{rank}(H_L(u))=mL$.
\end{definition}
In this paper, we propose a data-driven MPC scheme to control discrete-time LTI systems, using no model knowledge but only measured data, where the input generating the data is persistently exciting of a sufficiently high order.
We assume that the system order, denoted by $n$, is known, but our theoretical results remain true if $n$ is replaced by an upper bound.
Further, we write $m$ and $p$ for the number of inputs and outputs, and we assume that the input-output behavior of the system can be explained via a (controllable and observable) minimal realization.

\begin{definition}\label{def:traj}
We say that an input-output sequence $\{u_k,y_k\}_{k=0}^{N-1}$ is a trajectory of an LTI system $G$, if there exists an initial condition $\bar{x}\in\mathbb{R}^n$ as well as a state sequence $\{x_k\}_{k=0}^N$ such that
\begin{align*}
x_{k+1}&=Ax_k+Bu_k,\>\>x_0=\bar{x},\\
y_k&=Cx_k+Du_k,
\end{align*}
for $k=0,\dots,N-1$, where $(A,B,C,D)$ is a minimal realization of $G$.
\end{definition}

The following result, originally proven in the context of behavioral systems theory~\cite{Willems05}, shows that a single trajectory of an unknown LTI system can be used to parametrize all trajectories of the system if the input is persistently exciting.

\begin{theorem}[\cite{Berberich20a}]\label{thm:hankel}
Suppose $\{u_k^d,y_k^d\}_{k=0}^{N-1}$ is a trajectory of an LTI system $G$, where $u^d$ is persistently exciting of order $L+n$.
Then, $\{\bar{u}_k,\bar{y}_k\}_{k=0}^{L-1}$ is a trajectory of $G$ if and only if there exists $\alpha\in\mathbb{R}^{N-L+1}$ such that
\begin{align}\label{eq:thm_hankel}
\begin{bmatrix}H_L(u^d)\\H_L(y^d)\end{bmatrix}\alpha=\begin{bmatrix}\bar{u}\\\bar{y}
\end{bmatrix}.
\end{align}
\end{theorem}
In this paper, we employ Theorem~\ref{thm:hankel} to set up a data-driven MPC scheme which practically exponentially stabilizes an input-output setpoint $(u^s,y^s)$ while the closed-loop trajectory satisfies pointwise-in-time input and output constraints $u_t\in\mathbb{U}$ and $y_t\in\mathbb{Y}$.
The setpoint $(u^s,y^s)$ is required to be an equilibrium of the unknown LTI system in the sense that the sequence $\{\bar{u}_k,\bar{y}_k\}_{k=0}^{{\color{blue}n}}$ with $(\bar{u}_k,\bar{y}_k)=(u^s,y^s)$ for all $k\in\mathbb{I}_{[0,{\color{blue}n}]}$ is a trajectory of $G$.
We abbreviate $u_n^s$ and $y_n^s$ as the column vectors containing $n$ times the input and output setpoint, respectively, and we assume $(u^s,y^s)\in\text{int}(\mathbb{U}\times\mathbb{Y})$.
Further, the constraint set $\mathbb{U}$ is assumed to be a compact convex polytope.
For simplicity, we consider hyperbox constraints on the output, which are w.l.o.g. assumed to be symmetric w.r.t. zero, i.e., $\mathbb{Y}=\left\{y\in\mathbb{R}^p\mid \lVert y\rVert_\infty\leq y_{\max}\right\}$, where $y_{\max}>0$.
It is straightforward to extend the results of this paper to constraints with different bounds $y_{\max}^i$, $i=1,\dots,p$, for each component of the output, e.g., by scaling all output measurements by $\frac{y_{\max}^i}{y_{\max}}$.
Moreover, we conjecture that similar results can be obtained for general polytopic output constraints.

\section{Robust data-driven MPC}\label{sec:MPC}

In this section, we propose a data-driven MPC scheme which uses Theorem~\ref{thm:hankel} to predict future trajectories based on a single noisy data trajectory.
In order to guarantee output constraint satisfaction despite the noisy data, the scheme uses a constraint tightening which ensures that the closed-loop output satisfies the desired constraints $y_t\in\mathbb{Y}$ for all $t$. 
We first present the MPC scheme in Section~\ref{sec:MPC_scheme} and thereafter, in Section~\ref{sec:MPC_constraints}, we describe the constraint tightening in more detail.

\subsection{Proposed MPC scheme}\label{sec:MPC_scheme}
In this paper, we consider the scenario that the output measurements are affected by additive noise (i) in the initial data used for prediction via Theorem~\ref{thm:hankel}, i.e., $\tilde{y}_k^d=y_k^d+\varepsilon_k^d$, similar to parametric uncertainty in model-based MPC, and (ii) in the online measurements used for feedback, i.e., $\tilde{y}_k=y_k+\varepsilon_k$, where the noise is bounded as 
$\lVert\varepsilon_k^d\rVert_\infty\leq\bar{\varepsilon}$, $\lVert\varepsilon_k\rVert_\infty\leq\bar{\varepsilon}$, for some constant $\bar{\varepsilon}>0$.
Clearly, if the output trajectory $y^d$ in Theorem~\ref{thm:hankel} is affected by noise, then the theorem statement is invalid and hence, it is non-trivial to provide guarantees for a data-driven MPC scheme based on Theorem~\ref{thm:hankel} with noisy data.
The same noise setting was considered in~\cite{Berberich19b}, where recursive feasibility and practical stability of a data-driven MPC scheme were proven for the case of no output constraints, i.e., $\mathbb{Y}=\mathbb{R}^p$.
In this paper, we enhance the scheme of~\cite{Berberich19b} by a constraint tightening which guarantees closed-loop output constraint satisfaction.
To ensure stability of the setpoint $(u^s,y^s)$, we consider the stage cost
\begin{align*}
\ell(u,y)\coloneqq\lVert u-u^s\rVert_R^2+\lVert y-y^s\rVert_Q^2,
\end{align*}
with weighting matrices $Q,R\succ0$.
To this end, given a noisy data trajectory $\{u_k^d,\tilde{y}_k^d\}_{k=0}^{N-1}$, past $n$ input-output measurements $(u_{[t-n,t-1]},\tilde{y}_{[t-n,t-1]})$ to specify initial conditions, as well as a desired input-output setpoint $(u^s,y^s)$, we define the following data-driven MPC scheme.
\begin{subequations}\label{eq:MPC}
\begin{align}\nonumber
J_L^*&\big(u_{[t-n,t-1]},\tilde{y}_{[t-n,t-1]}\big)=\\\nonumber
\underset{\substack{\alpha(t),\sigma(t)\\\bar{u}(t),\bar{y}(t)}}{\min}&\sum_{k=0}^{L-1}\ell\left(\bar{u}_k(t),\bar{y}_k(t)\right)+\lambda_\alpha\bar{\varepsilon}\lVert\alpha(t)\rVert_2^2+\lambda_\sigma\lVert\sigma(t)\rVert_2^2\\
\label{eq:MPC1} s.t.\>\> &\>\begin{bmatrix}
\bar{u}(t)\\\bar{y}(t)+\sigma(t)\end{bmatrix}=\begin{bmatrix}H_{L+n}\left(u^d\right)\\H_{L+n}\left(\tilde{y}^d\right)\end{bmatrix}\alpha(t),\\\label{eq:MPC2}
&\>\begin{bmatrix}\bar{u}_{[-n,-1]}(t)\\\bar{y}_{[-n,-1]}(t)\end{bmatrix}=\begin{bmatrix}u_{[t-n,t-1]}\\\tilde{y}_{[t-n,t-1]}\end{bmatrix},\\\label{eq:MPC3}
&\>\begin{bmatrix}\bar{u}_{[L-n,L-1]}(t)\\\bar{y}_{[L-n,L-1]}(t)\end{bmatrix}=\begin{bmatrix}u^s_n\\y^s_n\end{bmatrix},\>\>\bar{u}_k(t)\in\mathbb{U},\\
%\label{eq:MPC4}
%&\>\bar{y}_k(t)\in\mathbb{Y}_k,\>\>k\in\mathbb{I}_{[0,L-1]},\\
\label{eq:MPC5}
&\>\>\lVert\sigma(t)\rVert_\infty\leq\bar{\varepsilon}\left(1+\lVert\alpha(t)\rVert_1\right),\\\label{eq:MPC6}
&\>\>\lVert\bar{y}_k(t)\rVert_\infty+a_{1,k}\lVert\bar{u}(t)\rVert_1+a_{2,k}\lVert\alpha(t)\rVert_1\\\nonumber
&\quad+a_{3,k}\lVert\sigma(t)\rVert_\infty+a_{4,k}\leq y_{\max},\>\>k\in\mathbb{I}_{[0,L-n-1]}.
\end{align}
\end{subequations}
As in previous works on data-driven MPC, Problem~\eqref{eq:MPC} relies on a Hankel matrix constraint~\eqref{eq:MPC1} in order to predict future system trajectories based on Theorem~\ref{thm:hankel}.
Further,~\eqref{eq:MPC2} ensures that the initial conditions are correctly initialized whereas~\eqref{eq:MPC3} is a terminal equality constraint on the extended (non-minimal) state 
\begin{align}\label{eq:extended_state}
\xi_t=\begin{bmatrix}u_{[t-n,t-1]}\\y_{[t-n,t-1]}\end{bmatrix},
\end{align}
similar to terminal equality constraints in model-based MPC~\cite{Rawlings17}.
In order to obtain recursive feasibility and stability guarantees despite noisy data, Problem~\eqref{eq:MPC} contains an $\ell_2$-regularization of the variable $\alpha(t)\in\mathbb{R}^{N-L-n+1}$ in the cost, as well as an additional slack variable $\sigma(t)\in\mathbb{R}^{p(L+n)}$ which renders the equality constraints feasible and whose norm is also penalized in the cost.
Note that, except for the constraint~\eqref{eq:MPC5}, all constraints in Problem~\eqref{eq:MPC} can be written as linear equality and inequality constraints.
As is discussed in more detail in~\cite{Berberich19b}, the non-convex constraint~\eqref{eq:MPC5} can be dropped in practice if $\lambda_\sigma$ is large enough, retaining the same theoretical guarantees.
In this case, Problem~\eqref{eq:MPC} is a strictly convex quadratic program of size similar to model-based MPC (compare~\cite{Berberich19b} for details) and can thus be solved efficiently.

Finally, Problem~\eqref{eq:MPC} contains the constraint tightening~\eqref{eq:MPC6} which ensures closed-loop output constraint satisfaction despite noisy measurements.
As will become clear later in the paper, the output prediction error induced by the noise depends on the size of the quantities $\bar{u}(t),\alpha(t),\sigma(t)$, which explains their occurrences in the constraint tightening.
Appropriate definitions of the constants $a_{i,k}$ to guarantee recursive feasibility and closed-loop constraint satisfaction are provided in Section~\ref{sec:MPC_constraints} (to be more precise, in Equations~\eqref{eq:constraint_tightening1} and~\eqref{eq:constraint_tightening2}).
Problem~\eqref{eq:MPC} is solved in an $n$-step fashion, as defined in Algorithm~\ref{alg:MPC_n_step}.

\begin{algorithm}
\begin{Algorithm}\label{alg:MPC_n_step}
\normalfont{\textbf{$n$-Step Data-Driven MPC Scheme}}
\begin{enumerate}
\item At time $t$, take the past $n$ measurements $u_{[t-n,t-1]}$, $\tilde{y}_{[t-n,t-1]}$ and solve~\eqref{eq:MPC}.
\item Apply the input sequence $u_{[t,t+n-1]}=\bar{u}_{[0,n-1]}^*(t)$ over the next $n$ time steps.
\item Set $t=t+n$ and go back to 1).
\end{enumerate}
\end{Algorithm}
\end{algorithm}
For the present MPC approach, considering a multi-step scheme instead of a standard ($1$-step) MPC scheme leads to stronger theoretical closed-loop guarantees, due to the joint occurence of noisy data in the prediction model~\eqref{eq:MPC1} and terminal equality constraints.
In particular, the closed-loop properties of the corresponding $1$-step MPC scheme (when appropriately modifying the constraints~\eqref{eq:MPC6}) hold only locally around the origin (compare~\cite[Remark 4]{Berberich19b} for details).
Throughout the paper, we denote the optimal solution of Problem~\eqref{eq:MPC} at time $t$ by $\bar{u}^*(t),\bar{y}^*(t),\alpha^*(t),\sigma^*(t)$.
Closed-loop values of the input and output at time $t$ are denoted by $u_t$ and $y_t$, respectively.

\subsection{Constraint tightening}\label{sec:MPC_constraints}

In this section, we define the constants $a_{i,k}$ involved in the constraint tightening~\eqref{eq:MPC6}, based on data and system parameters.
First, we assume knowledge of a controllability constant $\Gamma>0$ such that for any initial trajectory $\{u_k,y_k\}_{k=-n}^{n-1}$ with $u_{[0,n-1]}=0$, there exists an input-output trajectory $\{\bar{u}_k,\bar{y}_k\}_{k=-n}^{2n-1}$ such that
\begin{align}
(\bar{u}_{[-n,-1]},\bar{y}_{[-n,-1]})&=(u_{[-n,-1]},y_{[-n,-1]}),\\
(\bar{u}_{[n,2n-1]},\bar{y}_{[n,2n-1]})&=(0,0),\\\label{eq:ass_ctrb}
\lVert \bar{u}_{[0,n-1]}\rVert_1&\leq\Gamma\left\lVert y_{[0,n-1]}\right\rVert_\infty.
\end{align}
Essentially, this means that, for any given initial input-output trajectory $\{u_k,y_k\}_{k=-n}^{-1}$, there exists an input that can be appended to this trajectory which steers the system to zero in $n$ steps.
Moreover, the norm of this input is bounded from above by the norm of the output that would result from the same initial trajectory when applying a zero input.
Since we consider only systems whose input-output behavior can be explained by a minimal realization (compare Definition~\ref{def:traj}), the system is controllable.
Therefore, the above condition is always satisfied for \emph{some} $\Gamma>0$.
Further, we assume that there exist (known) constants $\rho_{k},k\in\mathbb{I}_{[n,L+n-1]},$ defined as
\begin{align}\nonumber
\rho_k\>\>=\quad&\max_{y\in\mathbb{R}^{(k+1)p}}\>\lVert y_k\rVert_{\infty}\\\label{eq:ass_obsv}
\text{s.t.}\>\>&\lVert y_{[0,n-1]}\rVert_\infty=1,\\\nonumber
&\{0,y_j\}_{j=0}^k\>\text{is a trajectory of $G$}.
\end{align}
The constants $\rho_k$ are related to observability of the underlying system.
According to its definition, $\rho_k$ is equal to the norm of the output of the unknown LTI system $G$ at time $k$, assuming that the input is zero and the initial output is norm-bounded by $1$.
In Section~\ref{sec:system_constants}, we describe how $\Gamma$ as well as $\rho_k$ can be estimated from measured data.
Furthermore, since the input and output constraint sets $\mathbb{U},\mathbb{Y}$ are compact, we can define 
\begin{align}\label{eq:xi_max_def}
\xi_{\max}\coloneqq\max_{\xi\in\mathbb{U}^n\times\mathbb{Y}^n}\lVert\xi\rVert_1.
\end{align}
Finally, we define
\begin{align}\label{eq:hankel_xi}
H_{u\xi}=\begin{bmatrix}H_{L+n}(u^d)\\H_1\left(\xi^d_{[0,N-L-n]}\right)\end{bmatrix},
\end{align}
where $\xi^d$ is the extended state\footnote{In order to construct this state at times $0$ through $n-1$, we require $n$ additional data points $\{u_k^d,y_k^d\}_{k=-n}^{-1}$.
For notational simplicity, we neglect this fact throughout the paper and simply assume that $N+n$ overall data points are available, the first $n$ of which are only used to construct $\{\xi_k^d\}_{k=0}^{n-1}$.} as defined in~\eqref{eq:extended_state}, corresponding to the measured data $\{u_k^d,y_k^d\}_{k=0}^{N-1}$.
We define the constant $c_{pe}\coloneqq\lVert H_{u\xi}^\dagger\rVert_1$, where $H_{u\xi}^\dagger$ is the Moore-Penrose inverse of $H_{u\xi}$.
Note that $H_{u\xi}$ does in general not have full row rank, even if the input is persistently exciting, since the state-space model with state $\xi$ is typically not controllable\footnote{According to~\cite[Corollary 2]{Willems05}, both persistence of excitation as well as controllability are required to ensure that $H_{u\xi}$ has full row rank.} (only under strict conditions, e.g., single-input single-output systems with exactly known order, compare~\cite{Persis19},~\cite[Lemma 3.4.7]{Goodwin14}).
Nevertheless, if the input is persistently exciting of order $L+2n$, then for any desired input $\{u_k\}_{k=0}^{L-1}$ and extended state $\xi_0$, there exists $\alpha\in\mathbb{R}^{N-L-n+1}$ such that
\begin{align}\label{eq:Huxi}
H_{u\xi}\alpha=\begin{bmatrix}u\\\xi_0\end{bmatrix},
\end{align}
compare~\cite{Markovsky08} for details.
A related constant $\lVert H_{ux}^\dagger\rVert_2^2$, where $x$ is the state in some minimal realization, was extensively studied in~\cite{Berberich19b}.
In this paper, we consider the matrix $H_{u\xi}$ instead of $H_{ux}$ since the extended state $\xi$ depends only on input-output measurements and a minimal state $x$ is not available in the present setting.
The constant $c_{pe}$ can be rendered arbitrarily small by choosing the norm of a sufficiently rich persistently exciting input $u^d$ sufficiently large.
In~\cite{Berberich19b}, in the absence of output constraints, it was shown that a small value of $c_{pe}$ corresponds to a large region of attraction and a small tracking error for the closed loop under the data-driven MPC scheme~\eqref{eq:MPC}.
Similarly, we will see in our main theoretical results that a small constant $c_{pe}$ reduces the conservatism of the proposed output constraint tightening.

Now, we are in the position to state the coefficients $a_{i,k}$, $i=1,\dots,4$, in~\eqref{eq:MPC6}.
For the first $n$ steps $k\in\mathbb{I}_{[0,n-1]}$, define
\begin{align}\label{eq:constraint_tightening1}
\begin{split}
&a_{1,k}=0,\>\>a_{2,k}=\bar{\varepsilon}a_{3,k},\\
&a_{3,k}=1+\rho_{n}^{\max},\>\>a_{4,k}=\bar{\varepsilon}\rho_{n}^{\max},
\end{split}
\end{align}
where $\rho_n^{\max}\coloneqq\max_{k\in\mathbb{I}_{[0,n-1]}}\rho_{n+k}$.
For $k\in\mathbb{I}_{[0,L-2n-1]}$, the coefficients $a_{i,k+n}$ are defined recursively as
\begin{align}\nonumber
a_{1,k+n}&=a_{1,k}+(a_{2,k}+a_{3,k}\bar{\varepsilon})c_{pe},\\\label{eq:constraint_tightening2}
a_{2,k+n}&=\bar{\varepsilon}a_{3,k+n},\\\nonumber
a_{3,k+n}&=1+\rho_{2n+k}+\Gamma(1+\rho_{L}^{\max})a_{1,k+n},\\\nonumber
a_{4,k+n}&=a_{4,k}+\bar{\varepsilon}\rho_{2n+k}+\bar{\varepsilon}a_{1,k+n}\Gamma\rho_{L}^{\max}\\\nonumber
&\quad+\bar{\varepsilon}a_{3,k}+(a_{2,k}+a_{3,k}\bar{\varepsilon})c_{pe}\xi_{\max},
\end{align}
where $\rho_L^{\max}\coloneqq\max_{k\in\mathbb{I}_{[0,n-1]}}\rho_{L+k}$.
According to the above definition, $a_{2,k}=\bar{\varepsilon}a_{3,k}$ for all $k\in\mathbb{I}_{[0,L-n-1]}$ and hence, $a_{2,k}$ is arbitrarily small if $\bar{\varepsilon}$ is arbitrarily small.
This implies that the same holds true for $a_{1,k}$ and $a_{4,k}$ and hence, $a_{3,k}$ becomes arbitrarily close to $1+\rho_{n+k}$ for $k\geq n$.
Thus, except for $a_{3,k}$, all constants involved in the constraint tightening can be rendered arbitrarily small for sufficiently small noise levels.
Finally, the slack variable $\sigma(t)$ becomes arbitrarily small for small noise levels due to the constraint~\eqref{eq:MPC5} and hence, the tightened constraints~\eqref{eq:MPC6} recover the nominal output constraints $\lVert\bar{y}_k(t)\rVert_\infty\leq y_{\max}$ if the noise level tends to zero.

In order to implement the MPC scheme, the coefficients $a_{i,k}$ need to be computed according to the above recursion.
This requires knowledge of the constants $\rho_{k},\Gamma,\xi_{\max},c_{pe}$.
Note that $\xi_{\max}$ can be computed directly using only the definition of the constraints $\mathbb{U},\mathbb{Y}$.
Similarly, using~\eqref{eq:hankel_xi}, the constant $c_{pe}$ can be (approximately) computed based on the data $\{u_k^d,\tilde{y}_k^d\}_{k=0}^{N-1}$, i.e.,
\begin{align*}
c_{pe}\approx\lVert H_{u\tilde{\xi}}\rVert_1,
\end{align*}
with $\tilde{\xi}_k=\begin{bmatrix}u_{[k-n,k-1]}^\top&\tilde{y}_{[k-n,k-1]}^\top\end{bmatrix}^\top$.
On the other hand, the quantities $\rho_{k}$ and $\Gamma$ cannot be computed without additional model knowledge.
In Section~\ref{sec:system_constants}, we propose methods to estimate these constants from measured data based on Theorem~\ref{thm:hankel}.

\section{Theoretical guarantees}\label{sec:theory}

In this section, we provide a theoretical analysis of closed-loop properties of the proposed MPC scheme with the constraint tightening described in Section~\ref{sec:MPC_constraints}.
First, we state the main assumptions in Section~\ref{sec:theory_ass} and thereafter, we prove recursive feasibility as well as closed-loop stability and constraint satisfaction in Sections~\ref{sec:theory_feas} and~\ref{sec:theory_stab}.

\subsection{Assumptions}\label{sec:theory_ass}

In the remainder of the paper, we consider the case $(u^s,y^s)=(0,0)$, i.e., the desired setpoint is the origin, and we comment on setpoints $(u^s,y^s)\neq(0,0)$ in Remark~\ref{rk:setpoint0}.
We make the assumption that the input generating the initial data is sufficiently rich in the following sense.

\begin{assumption}\label{ass:pe}
The input $u^d$ of the data trajectory is persistently exciting of order $L+2n$.
\end{assumption}

Note that Assumption~\ref{ass:pe} requires persistence of excitation of order $L+2n$, although Theorem~\ref{thm:hankel} assumes only an order of $L+n$.
This is due to the fact that the length of the trajectories predicted via Problem~\eqref{eq:MPC} is $L+n$, due to the initial trajectory in~\eqref{eq:MPC2}.
Moreover, we require that the prediction horizon $L$ is at least twice as long as the system order.

\begin{assumption}\label{ass:L_geq_2n}
The prediction horizon satisfies $L\geq2n$.
\end{assumption}

In order to provide theoretical guarantees on closed-loop output constraint satisfaction, we require information on the prediction error induced by the additive noise in the ``prediction model''~\eqref{eq:MPC1}, i.e., on the difference between the predicted solution $\bar{y}^*(t)$ and the actual trajectory $\hat{y}_{[t,t+L-1]}$ which would result from an (open-loop) application of $\bar{u}^*(t)$.
To be more precise, $\hat{y}$ is defined as
\begin{align}\label{eq:y_hat_def}
\hat{y}_{t+k}=CA^kx_t+\sum_{j=0}^{k-1}CA^{k-1-j}B\bar{u}^*_j(t)+D\bar{u}^*_k(t),
\end{align}
where $x$ is the state in some minimal realization $(A,B,C,D)$.
In~\cite{Berberich19b}, the following bound on this difference is derived for the proposed MPC scheme.
\begin{lemma}\label{lem:pred_error}
\cite[Lemma 2]{Berberich19b} If~\eqref{eq:MPC} is feasible at time $t$, then the following holds for all $k\in\mathbb{I}_{[0,L-1]}$
\begin{align}\label{eq:lem_pred_error}
\begin{split}
\lVert\hat{y}_{t+k}-\bar{y}_k^*(t)\rVert_\infty\leq&\bar{\varepsilon}\rho_{n+k}+\bar{\varepsilon}(1+\rho_{n+k})\lVert\alpha^*(t)\rVert_1\\
&+(1+\rho_{n+k})\lVert\sigma^*(t)\rVert_\infty.
\end{split}
\end{align}
\end{lemma}
Lemma~\ref{lem:pred_error} shows that the prediction error depends not only on the noise bound $\bar{\varepsilon}$ and the slack variable $\sigma$, but also on the weighting vector $\alpha$.
As we will see in the recursive feasibility proof in Section~\ref{sec:theory_feas}, the candidate solution for $\alpha$ in turn depends on the optimal input $\bar{u}^*(t)$, which thus explains why all of these variables are included in the tightened constraints~\eqref{eq:MPC6}.
In~\cite{Berberich19b}, it is shown that~\eqref{eq:lem_pred_error} holds as long as $\rho_k$ is an upper bound on $\lVert CA^k\Phi^\dagger\rVert_\infty$ for $k\in\mathbb{I}_{[n,L+n-1]}$, where $\Phi$ denotes the observability matrix $\Phi=\begin{bmatrix}C^\top&A^\top C^\top&\dots&(A^{n-1})^\top C^\top\end{bmatrix}^\top$.
It is not difficult to see that, by definition of the induced matrix norm, $\rho_k$ as defined in~\eqref{eq:ass_obsv} is equal to $\lVert CA^k\Phi^\dagger\rVert_\infty$, which thus implies that~\eqref{eq:lem_pred_error} holds.

\subsection{Recursive feasibility}\label{sec:theory_feas}
In this section, we prove that, if the proposed data-driven MPC scheme is feasible at time $t$, then it is feasible at time $t+n$, provided that the noise bound $\bar{\varepsilon}$ is sufficiently small.
While this was proven in~\cite{Berberich19b} for Problem~\eqref{eq:MPC} without any constraints on the output, i.e., without~\eqref{eq:MPC6}, it is an essential contribution of the present paper to prove the same recursive feasibility property for the scheme with tightened output constraints~\eqref{eq:MPC6}.
First, let us define an input-output-to-state-stability (IOSS) Lyapunov function $W(\xi)=\lVert \xi\rVert_P^2$ with $P\succ0$, where $\xi$ is the extended state defined in~\eqref{eq:extended_state} corresponding to the (detectable) system realization $(\tilde{A},\tilde{B},\tilde{C},\tilde{D})$, such that
\begin{align*}
W(\tilde{A}\xi+\tilde{B}u)-W(\xi)\leq-\frac{1}{2}\lVert \xi\rVert_2^2+c_1\lVert u\rVert_2^2+c_2\lVert y\rVert_2^2,
\end{align*}
for all $u,\xi,$ and $y=\tilde{C}\xi+\tilde{D}u$ (compare~\cite{cai2008input}).
As a Lyapunov function candidate for the closed loop under the proposed MPC scheme we consider the sum of the optimal cost of~\eqref{eq:MPC} and an IOSS Lyapunov function $W$, i.e.,
\begin{align*}
V_t\coloneqq J_L^*(u_{[t-n,t-1]},\tilde{y}_{[t-n,t-1]})+\gamma W(\xi_t),
\end{align*}
for some $\gamma>0$.

\begin{proposition}\label{prop:rec_feas}
Suppose Assumptions~\ref{ass:pe} and~\ref{ass:L_geq_2n} hold.
Then, for any $V_{ROA}>0$, there exists an $\bar{\varepsilon}_0>0$ such that for all $\bar{\varepsilon}\leq\bar{\varepsilon}_0$, if $V_t\leq V_{ROA}$ for some $t\geq0$, then Problem~\eqref{eq:MPC} is feasible at time $t+n$ for the resulting closed loop.
\end{proposition}
\begin{proof}
First, we recall the candidate solution at time $t+n$ used to prove recursive feasibility in~\cite{Berberich19b} without the tightened constraints~\eqref{eq:MPC6}.
For $k\in\mathbb{I}_{[-n,L-2n-1]}$, the input candidate is chosen as the shifted previously optimal solution, i.e., $\bar{u}_k'(t+n)=\bar{u}_{k+n}^*(t)$.
The output candidate is chosen as $\bar{y}_{[-n,-1]}'(t+n)=\tilde{y}_{[t,t+n-1]}$ and $\bar{y}_k'(t+n)=\hat{y}_{t+n+k}$ for $k\in\mathbb{I}_{[0,L-2n-1]}$, where $\hat{y}$ denotes the output trajectory resulting from an open-loop application of the input $\bar{u}^*(t)$, compare~\eqref{eq:y_hat_def}.
It follows from controllability that the above input candidate can be extended such that $\bar{u}_{[L-2n,L-n-1]}'(t+n)$ steers the system to zero in $n$ steps.
If the noise bound $\bar{\varepsilon}$ is sufficiently small, then $\hat{y}_{t+k}$ becomes arbitrarily small for $k\in\mathbb{I}_{[L-n,L-1]}$ due to~\eqref{eq:lem_pred_error} together with $\bar{y}^*_{[L-n,L-1]}(t)=0$.
Therefore, for $\bar{\varepsilon}$ sufficiently small, the above-defined input $\bar{u}'_{[L-2n,L-n-1]}(t+n)$ satisfies the input constraints in~\eqref{eq:MPC3} (recall that we assumed $0\in\text{int}(\mathbb{U})$).
At time instants $k\in\mathbb{I}_{[L-2n,L-n-1]}$, the output candidate $\bar{y}_k'(t+n)$ is chosen as the output trajectory resulting from this input.
Further, the weighting vector candidate $\alpha'(t+n)$ is chosen as
\begin{align}\label{eq:prop_proof_alpha_candidate}
\alpha'(t+n)=H_{u\xi}^\dagger\begin{bmatrix}\bar{u}'(t+n)\\\xi_t\end{bmatrix},
\end{align}
where $H_{u\xi}^\dagger$ is the Moore-Penrose inverse of $H_{u\xi}$, compare also~\eqref{eq:Huxi}.
The slack variable is defined as $\sigma'(t+n)=H_{L+n}(\tilde{y}^d)\alpha'(t+n)-\bar{y}'(t+n)$.
This implies that~\eqref{eq:MPC1}-\eqref{eq:MPC3} hold.
Moreover, analogously to~\cite{Berberich19b}, it can be shown that this candidate solution also satisfies~\eqref{eq:MPC5}.
It thus only remains to show that~\eqref{eq:MPC6} holds.
Note that, by definition, the optimal solution at time $t$ satisfies
\begin{align}\label{eq:prop_rec_feas_proof1}
\lVert\bar{y}_k^*(t)\rVert_\infty+&a_{1,k}\lVert\bar{u}^*(t)\rVert_1+a_{2,k}\lVert\alpha^*(t)\rVert_1\\\nonumber
&+a_{3,k}\lVert\sigma^*(t)\rVert_\infty+a_{4,k}\leq y_{\max}.
\end{align}
We bound now each separate component appearing in the constraint~\eqref{eq:MPC6} at time $t+n$ to show that it is satisfied for $k\in\mathbb{I}_{[0,L-2n-1]}$ by the candidate solution defined above.
First, it holds for the output that
\begin{align}\nonumber
&\lVert\bar{y}_k'(t+n)\rVert_\infty\leq\lVert\bar{y}_{k+n}^*(t)\rVert_\infty+\lVert\hat{y}_{t+n+k}-\bar{y}_{k+n}^*(t)\rVert_\infty\\\label{eq:prop_proof_y}
&\stackrel{\eqref{eq:lem_pred_error},\eqref{eq:prop_rec_feas_proof1}}{\leq}
y_{\max}-a_{1,k+n}\lVert\bar{u}^*(t)\rVert_1-a_{2,k+n}\lVert\alpha^*(t)\rVert_1\\\nonumber
&-a_{3,k+n}\lVert\sigma^*(t)\rVert_\infty-a_{4,k+n}+\bar{\varepsilon}(1+\rho_{2n+k})\lVert\alpha^*(t)\rVert_1\\\nonumber
&+(1+\rho_{2n+k})\lVert\sigma^*(t)\rVert_\infty+\bar{\varepsilon}\rho_{2n+k}.
\end{align}
We obtain for the input candidate that
\begin{align*}
\lVert\bar{u}'(t+n)\rVert_1=&\lVert\bar{u}_{[0,L-n-1]}^*(t)\rVert_1\\
&+\lVert\bar{u}_{[L-2n,L-n-1]}'(t+n)\rVert_1.
\end{align*}
Due to the terminal equality constraint $\bar{u}^*_{[L-n,L-1]}(t)=0$ as well as the definition of $\hat{y}$ in~\eqref{eq:y_hat_def}, we can use~\eqref{eq:ass_ctrb} to bound the second term on the right-hand side as
\begin{align*}
\lVert\bar{u}_{[L-2n,L-n-1]}'(t+n)\rVert_1\stackrel{\eqref{eq:ass_ctrb}}{\leq}\Gamma\lVert\hat{y}_{[t+L-n,t+L-1]}\rVert_\infty.
\end{align*}
Using now the bound~\eqref{eq:lem_pred_error} together with the fact that $\bar{y}_k^*(t)=0$ for $k\in\mathbb{I}_{[L-n,L-1]}$ according to the terminal equality constraint~\eqref{eq:MPC3}, it follows that
\begin{align}\label{eq:prop_proof_u}
\lVert&\bar{u}'_{[L-2n,L-n-1]}(t+n)\rVert_1\leq 
\Gamma\big(\bar{\varepsilon}\rho_{L}^{\max}\\\nonumber
&+\bar{\varepsilon}(1+\rho_{L}^{\max})\lVert\alpha^*(t)\rVert_1+(1+\rho_{L}^{\max})\lVert\sigma^*(t)\rVert_\infty\big),
\end{align}
where $\rho_L^{\max}=\max_{k\in\mathbb{I}_{[0,n-1]}}\rho_{L+k}$.
Further, due to~\eqref{eq:prop_proof_alpha_candidate}, $\alpha'(t+n)$ can be bounded as
\begin{align}\label{eq:prop_proof_alpha}
\lVert\alpha'(t+n)\rVert_1\leq&c_{pe}(\lVert\bar{u}'(t+n)\rVert_1+\underbrace{\lVert\xi_t\rVert_1}_{\leq\xi_{\max}}),
\end{align}
with $c_{pe}=\lVert H_{u\xi}^\dagger\rVert_1$.
Finally, it is shown in~\cite{Berberich19b} that the slack variable candidate satisfies the constraint~\eqref{eq:MPC5}, i.e.,
\begin{align}\label{eq:prop_proof_sigma}
\lVert\sigma'(t+n)\rVert_\infty\leq&\bar{\varepsilon}(1+\lVert\alpha'(t+n)\rVert_1).
\end{align}
It follows from straightforward algebraic manipulations, using the definition of the coefficients $a_{i,k}$ in~\eqref{eq:constraint_tightening1} and~\eqref{eq:constraint_tightening2} as well as the bounds~\eqref{eq:prop_proof_y}, \eqref{eq:prop_proof_u}, \eqref{eq:prop_proof_alpha}, and~\eqref{eq:prop_proof_sigma}, that the candidate solution satisfies
\begin{align*}
\lVert\bar{y}_k'(t+n)\rVert_\infty+&a_{1,k}\lVert\bar{u}'(t+n)\rVert_1+a_{2,k}\lVert\alpha'(t+n)\rVert_1\\
&+a_{3,k}\lVert\sigma'(t+n)\rVert_\infty+a_{4,k}\leq y_{\max},
\end{align*}
i.e., the constraint~\eqref{eq:MPC6} holds for $k\in\mathbb{I}_{[0,L-2n-1]}$.
If $\bar{\varepsilon}$ is sufficiently small, then $\bar{y}_k'(t+n)$ is arbitrarily small for $k\in\mathbb{I}_{[L-2n,L-n-1]}$ (recall its definition at the beginning of the proof) and also the other components appearing on the left-hand side of the constraint~\eqref{eq:MPC6} become arbitrarily small (compare the discussion below~\eqref{eq:constraint_tightening2}).
Hence, since $y_{\max}>0$, the constraint~\eqref{eq:MPC6} holds also for $k\in\mathbb{I}_{[L-2n,L-n-1]}$ which thus concludes the proof.
\end{proof}

Proposition~\ref{prop:rec_feas} shows that the proposed $n$-step MPC scheme is $n$-step feasible in the sense that, for a given Lyapunov function sublevel set $V_t\leq V_{ROA}$ with some $V_{ROA}>0$, there exists a noise bound $\bar{\varepsilon}$ sufficiently small such that the scheme is feasible at time $t+n$.
In Section~\ref{sec:theory_stab} we prove that, in addition, the sublevel set $V_t\leq V_{ROA}$ is positively invariant, which then implies recursive feasibility of the MPC scheme in the classical sense.
The set $V_t\leq V_{ROA}$ can be seen as the guaranteed region of attraction of the closed loop under the proposed MPC scheme.
In order to ensure a larger region of attraction $V_{ROA}$, the admissible noise bound $\bar{\epsilon}$ needs to be smaller~\cite{Berberich19b}.

The proof of Proposition~\ref{prop:rec_feas} utilizes the same candidate solution as in~\cite{Berberich19b}.
While it is shown in~\cite{Berberich19b} that this candidate solution satisfies the constraints~\eqref{eq:MPC1}-\eqref{eq:MPC5}, Proposition~\ref{prop:rec_feas} proves that it additionally satisfies the tightened constraints~\eqref{eq:MPC6}, which is the main technical contribution of this paper.

\begin{remark}\label{rk:setpoint0}
All results in this paper hold true qualitatively for non-zero setpoints $(u^s,y^s)\neq(0,0)$, but the guarantees deteriorate quantitatively.
To be more precise, for non-zero setpoints, closed-loop stability and constraint satisfaction hold true under suitable assumptions on system and design parameters, but the maximal noise bound for which they are guaranteed decreases (compare~\cite[Remark 5]{Berberich19b} for details).
Moreover, the implementation of the constraint tightening~\eqref{eq:MPC6} changes slightly when considering $(u^s,y^s)\neq(0,0)$.
This is due to the fact that, in this case, the controllability bound~\eqref{eq:ass_ctrb} leads to
\begin{align*}
\lVert\bar{u}_{[0,n-1]}-u^s_n\rVert_1\leq\Gamma\lVert y_{[0,n-1]}-y^s_n\rVert_\infty.
\end{align*}
When modifying the bound~\eqref{eq:prop_proof_alpha} accordingly, an additional additive term on the right-hand side of~\eqref{eq:prop_proof_alpha} needs to be introduced, depending on the setpoint $(u^s,y^s)$.
In this paper, we do not consider the case $(u^s,y^s)\neq(0,0)$ since, compared to the other terms involved in the constraint tightening, the effect of a non-zero setpoint is relatively small and considering it complicates some of the arguments.
\end{remark}

\subsection{Constraint satisfaction and stability}\label{sec:theory_stab}

The following result shows that the proposed MPC scheme renders the origin practically exponentially stable w.r.t. the noise bound $\bar{\varepsilon}$, and that the closed-loop output satisfies the constraints $y_t\in\mathbb{Y}$ for all $t\geq0$.

\begin{theorem}\label{thm:MPC_stab}
Suppose Assumptions~\ref{ass:pe} and~\ref{ass:L_geq_2n} hold.
Then, for any $V_{ROA}>0$, there exist constants $\underline{\lambda}_\alpha,\overline{\lambda}_\alpha,\underline{\lambda}_\sigma,\overline{\lambda}_\sigma>0$ such that, for all $\lambda_\alpha,\lambda_\sigma$ satisfying
\begin{align}\label{eq:thm_bounds1}
\begin{split}
\underline{\lambda}_\alpha\leq\lambda_\alpha\leq\overline{\lambda}_\alpha,\quad\underline{\lambda}_\sigma\leq\lambda_\sigma\leq\overline{\lambda}_\sigma,
\end{split}
\end{align}
there exist constants $\bar{\varepsilon}_0,\bar{c}_{pe}>0$, as well as a continuous, strictly increasing function $\beta:[0,\bar{\varepsilon}_0]\to[0,V_{ROA}]$ with $\beta(0)=0$, such that, for all $\bar{\varepsilon},c_{pe}$ satisfying
\begin{align}\label{eq:thm_bounds2}
\bar{\varepsilon}\leq\min\left\{\bar{\varepsilon}_0,\frac{\bar{c}_{pe}}{c_{pe}}\right\},
\end{align}
for any initial condition satisfying $V_0\leq V_{ROA}$, the closed loop of the $n$-step MPC scheme satisfies the constraints, i.e., $y_t\in\mathbb{Y}$ for all $t\geq0$, and $V_t$ converges exponentially to $V_t\leq\beta(\bar{\varepsilon})$.
\end{theorem}
\begin{proof}
In~\cite{Berberich19b}, it is shown that the set $V_t\leq V_{ROA}$ is robustly positively invariant and that $V_t$ converges exponentially to $V_t\leq\beta(\bar{\varepsilon})$, using the same candidate solution as in Proposition~\ref{prop:rec_feas}, but for an MPC scheme without the tightened output constraint~\eqref{eq:MPC6}.
As Proposition~\ref{prop:rec_feas} shows, also the tightened constraints are satisfied by this candidate solution and hence, we can apply the exact same arguments as in~\cite{Berberich19b} to conclude recursive feasibility as well as exponential convergence to $V_t\leq\beta(\bar{\varepsilon})$.
To prove closed-loop output constraint satisfaction, note that $\hat{y}_{t+k}=y_{t+k}$ for $k\in\mathbb{I}_{[0,n-1]}$ due to the $n$-step MPC scheme, with $\hat{y}$ from~\eqref{eq:y_hat_def}.
Hence,~\eqref{eq:lem_pred_error} implies
\begin{align}\label{eq:thm_proof_stab1}
\lVert y_{t+k}-\bar{y}_k^*(t)\rVert_\infty&\leq\bar{\varepsilon}(1+\rho_{n}^{\max})\lVert\alpha^*(t)\rVert_1\\\nonumber
&\quad+(1+\rho_{n}^{\max})\lVert\sigma^*(t)\rVert_\infty+\bar{\varepsilon}\rho_{n}^{\max},
\end{align}
for any $k\in\mathbb{I}_{[0,n-1]}$, where $\rho_n^{\max}=\max_{k\in\mathbb{I}_{[0,n-1]}}\rho_{n+k}$.
Further, since the optimal solution satisfies~\eqref{eq:MPC6}, it follows from the definition of the coefficients $a_{i,k}$, $i\in\mathbb{I}_{[1,4]},k\in\mathbb{I}_{[0,n-1]}$, in~\eqref{eq:constraint_tightening1} that
\begin{align}\label{eq:thm_proof_stab2}
\begin{split}
\lVert\bar{y}_k^*(t)\rVert_\infty&\leq y_{\max}-\bar{\varepsilon}(1+\rho_{n}^{\max})\lVert\alpha^*(t)\rVert_1\\
&-(1+\rho_{n}^{\max})\lVert\sigma^*(t)\rVert_\infty-\bar{\varepsilon}\rho_{n}^{\max},
\end{split}
\end{align}
for any $k\in\mathbb{I}_{[0,n-1]}$.
Combining~\eqref{eq:thm_proof_stab1} and~\eqref{eq:thm_proof_stab2}, we thus obtain
\begin{align*}
\lVert y_{t+k}\rVert_\infty\leq\lVert y_{t+k}-\bar{y}_k^*(t)\rVert_\infty+\lVert\bar{y}_k^*(t)\rVert_\infty\leq y_{\max},
\end{align*}
i.e., the closed-loop output satisfies the constraints.
\end{proof}

Except for output constraint satisfaction $y_t\in\mathbb{Y}$, Theorem~\ref{thm:MPC_stab} follows from Proposition~\ref{prop:rec_feas} and the results in~\cite{Berberich19b}.
It is the key contribution of the present paper to suggest a suitable output constraint tightening, compare~\eqref{eq:MPC6} with parameters~\eqref{eq:constraint_tightening1} and~\eqref{eq:constraint_tightening2}, and to prove recursive feasibility of this constraint tightening (Proposition~\ref{prop:rec_feas}) as well as closed-loop constraint satisfaction (Theorem~\ref{thm:MPC_stab}).
A detailed discussion of the conditions on the system parameters $\lambda_{\alpha},\lambda_{\sigma}$ and the data parameters $\bar{\varepsilon}$ and $c_{pe}$, and their relation to the region of attraction $V_0\leq V_{ROA}$ and the tracking error described via $\beta(\bar{\varepsilon})$, can be found in~\cite{Berberich19b}.
Loosely speaking, it is shown therein that closed-loop performance improves if either the noise bound decreases or persistence of excitation increases quantitatively (i.e., $c_{pe}$ decreases).

\section{Estimation of system constants}\label{sec:system_constants}
In order to define the tightened output constraints~\eqref{eq:MPC6} and thus, to implement the proposed MPC scheme in practice, appropriate values for the coefficients $a_{i,k}$ as in~\eqref{eq:constraint_tightening1} and~\eqref{eq:constraint_tightening2} need to be computed.
This requires knowledge of the system constants $\Gamma$ and $\rho_k$ defined in~\eqref{eq:ass_ctrb} and~\eqref{eq:ass_obsv}, respectively.
In this section, we employ Theorem~\ref{thm:hankel} to derive a purely data-driven estimation procedure for these system constants.
Similar to~\cite{Koch20}, which verifies integral quadratic constraints from measured data, we use Theorem~\ref{thm:hankel} to optimize over all system trajectories based on a single measured trajectory in order to compute the desired constants, i.e., to verify a quantitative controllability and observability property.

While the stability and constraint satisfaction results of the previous section hold true if the data are affected by noise (assuming that upper bounds on $\rho_k$ and $\Gamma$ are available), we assume throughout this section that a noise-free input-output trajectory $\{u_k^d,y_k^d\}$ of the unknown system is available.
An extension of the estimation procedures presented in this section to the case of noisy measurements is an interesting issue for future research.

\subsection{Controllability}
In the following, we derive a purely data-driven optimization problem to compute a constant $\Gamma>0$ according to~\eqref{eq:ass_ctrb}.
To be more precise, $\Gamma$ is equal to the optimal value of the optimization problem
\begin{align}\nonumber
\max_{u,y,\alpha}\min_{\bar{u},\bar{y},\bar{\alpha}}\quad&\lVert\bar{u}_{[0,n-1]}\rVert_1\\\nonumber
s.t.\quad&(\bar{u}_{[-n,-1]},\bar{y}_{[-n,-1]})=(u_{[-n,-1]},y_{[-n,-1]}),\\\nonumber
&(\bar{u}_{[n,2n-1]},\bar{y}_{[n,2n-1]})=(0,0),\\\label{eq:Gamma}
&\lVert y_{[0,n-1]}\rVert_\infty\leq1,\>\>u_{[0,n-1]}=0,\\\nonumber
&\begin{bmatrix}H_{3n}(u^d)\\H_{3n}(y^d)\end{bmatrix}\bar{\alpha}=\begin{bmatrix}\bar{u}\\\bar{y}
\end{bmatrix},\begin{bmatrix}H_{2n}(u^d)\\H_{2n}(y^d)\end{bmatrix}\alpha=\begin{bmatrix}u\\y
\end{bmatrix}.
\end{align}
In Problem~\eqref{eq:Gamma}, we impose $\lVert y_{[0,n-1]}\rVert_\infty\leq1$ instead of $\lVert y_{[0,n-1]}\rVert_\infty=1$ since the inequality constraint is always satisfied with equality due to the maximization in $y$, and the constraint $\lVert y_{[0,n-1]}\rVert_\infty\leq1$ is convex.
We proceed in solving~\eqref{eq:Gamma} as follows:
The set of all feasible $(u,y)$ in~\eqref{eq:Gamma} is a compact convex polytope $\mathbb{Z}$ and hence, instead of maximizing over all $(u,y)$, it suffices to solve the minimization problem for all vertices of this polytope.
To be more precise, denote the set of vertices of $\mathbb{Z}$ by $\mathbb{Z}_v=\{(u^i,y^i), i\in\mathbb{I}_{[1,l]}\}$ and note that $\mathbb{Z}_v$ can be conveniently computed from data, e.g., using the MPT3-toolbox~\cite{Herceg13}.
We compute $\Gamma$ as the optimal value of
\begin{align}\label{eq:Gamma_LP}
\max_{i\in\mathbb{I}_{[1,l]}}\min_{\bar{u},\bar{y},\bar{\alpha}}\quad&\lVert\bar{u}_{[0,n-1]}\rVert_1\\\nonumber
s.t.\quad&(\bar{u}_{[-n,-1]},\bar{y}_{[-n,-1]})=(u^i_{[-n,-1]},y^i_{[-n,-1]}),\\\nonumber
&(\bar{u}_{[n,2n-1]},\bar{y}_{[n,2n-1]})=(0,0),\\\nonumber
&\begin{bmatrix}H_{3n}(u^d)\\H_{3n}(y^d)\end{bmatrix}\bar{\alpha}=\begin{bmatrix}\bar{u}\\\bar{y}\end{bmatrix}.
\end{align}
Note that, for any fixed $i$, the inner minimization problem in~\eqref{eq:Gamma_LP} is a linear program (LP) which can be solved efficiently using standard solvers.
Hence, to compute the constant $\Gamma$, we need to (i) compute the vertex set $\mathbb{Z}_v$ and (ii) solve a single LP for each element of $\mathbb{Z}_v$.
Due to the computation of the vertices $\mathbb{Z}_v$, the complexity of the above approach scales exponentially with the number of inputs and outputs and with the system dimension.
Nevertheless, it remains practical for medium-sized problems, due to the availability of efficient LP solvers.
Moreover, the computation is carried out once offline, i.e., it does not increase the online complexity of the proposed MPC scheme.

\subsection{Observability}
Next, we compute the constants $\rho_k$, $k\in\mathbb{I}_{[n,L+n-1]}$, as in~\eqref{eq:ass_obsv} from data.
Applying Theorem~\ref{thm:hankel}, the optimization problem~\eqref{eq:ass_obsv} can be reformulated as
\begin{align}\nonumber
\rho_k\>\>=\quad&\max_{y,\alpha}\>\lVert y_k\rVert_\infty\\
\label{eq:rho_k}
\text{s.t.}\>\>&\lVert y_{[0,n-1]}\rVert_\infty\leq1,\\\nonumber
&\begin{bmatrix}H_{k+1}(u^d)\\H_{k+1}(y^d)\end{bmatrix}\alpha=\begin{bmatrix}0\\y_{[0,k]}
\end{bmatrix}.
\end{align}
As in~\eqref{eq:Gamma}, we replace $\lVert y_{[0,n-1]}\rVert_\infty=1$ by $\lVert y_{[0,n-1]}\rVert_\infty\leq1$ since the constraint is always active for the optimal solution.
The maximization of $\lVert y_k\rVert_{\infty}$ can be reformulated as a linear objective via additional integer variables and hence, Problem~\eqref{eq:rho_k} can be directly solved using a mixed integer programming solver, which is, e.g., provided by MOSEK~\cite{MOSEK15}.
For systems with one output, i.e., $p=1$, Problem~\eqref{eq:rho_k} is an LP.

\section{Example}\label{sec:example}
In this section, we apply the proposed MPC scheme to the example from~\cite{Terzi19}.
The goal is to control an unknown LTI system with transfer function
\begin{align*}
G(z)=0.01\frac{2z^2+6.1z+1.1}{z^3-2.1z^2+1.5z-0.3}
\end{align*}
such that closed-loop input and output constraint satisfaction is guaranteed for the constraint sets $\mathbb{U}=\mathbb{Y}=[-10,10]$, i.e., $y_{\max}=10$.
We produce data $\{u_k^d,\tilde{y}_k^d\}_{k=0}^{N-1}$ by samplying the input uniformly from $\mathbb{U}$, where the data length is $N=1000$.
The output is affected by measurement noise with bound $\bar{\varepsilon}=0.0001$.
In order to study the conservatism of the tightened output constraints, we consider the stage cost $\ell(u,y)=-y$, i.e., the MPC objective is maximization of the output $y$.
Further, we choose $\lambda_\alpha\varepsilon=1$ as well as $\lambda_\sigma=100$.
Note that, since the stage cost function is not positive definite in the input and output, the stability guarantees derived in Theorem~\ref{thm:MPC_stab} do not hold in general.
Nevertheless, it is readily seen that recursive feasibility and constraint satisfaction are not affected since a terminal equality constraint as in~\eqref{eq:MPC3} is included.
We consider a terminal equality constraint with the setpoint $(u^s,y^s)=(5,5)$ and a prediction horizon $L=10$.
The system constants $\Gamma$ and $\rho_k$ for the constraint tightening are computed with the procedure described in Section~\ref{sec:system_constants} and they are indeed equal to the corresponding model-based values.

\begin{figure}[h!]
\begin{center}
\includegraphics[width=0.516\textwidth]{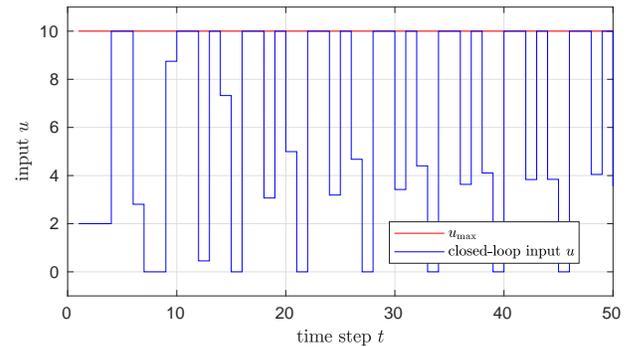}
\end{center}
\caption{Closed-loop input computed via the proposed MPC scheme.}
\label{fig:ex_u}
\end{figure}

\begin{figure}[h!]
\begin{center}
\includegraphics[width=0.516\textwidth]{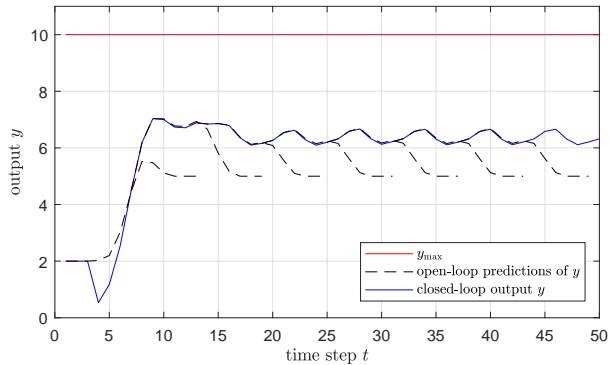}
\end{center}
\caption{Open-loop predictions and closed-loop output trajectory resulting from an application of the proposed MPC scheme.}
\label{fig:ex_y}
\end{figure}

The closed-loop input of the proposed data-driven $n$-step MPC scheme is displayed in Figure~\ref{fig:ex_u}, where it can be seen that the input operates on the boundary of the constraints $u_{\max}=10$ at several time steps.
Moreover, Figure~\ref{fig:ex_y} shows open-loop predictions at several time instants as well as the closed-loop output under the proposed MPC scheme.
Due to the tightened constraints~\eqref{eq:MPC6}, the output does not approach the boundary of the constraints but oscillates around $y=7$.
A similar behavior can be observed for the application of the MPC scheme in~\cite{Terzi19}.
Notably, the noise level that can be handled in~\cite{Terzi19} for the present example is considerably larger than the one considered above.
For larger noise levels, the proposed MPC scheme is initially infeasible since $a_{4,k}\geq y_{\max}$ already for small values of $k$.
This is due to the fact that, in our proof, several conservative bounds are used.
In particular, for bounding $\alpha'(t+n)$ in the proof of Proposition~\ref{prop:rec_feas}, we bound a product involving $H_{u\xi}^\dagger$ by $c_{pe}$ and we replace $\lVert\xi_t\rVert_1$ by $\xi_{\max}$, cf.~\eqref{eq:xi_max_def}, which is the main source of conservatism.
Nevertheless, the approach presented in this paper leads to end-to-end guarantees from noisy data of finite length to closed-loop output constraint satisfaction under mild assumptions.
Compared to~\cite{Terzi19}, our method is more general (i.e., it requires less assumptions) but more conservative, and it is very simple to apply since only a measured data trajectory as well as scalar estimates of the system constants $\Gamma$ and $\rho_k$ are required for its implementation.

\section{Conclusion}\label{sec:conclusion}

In this paper, we augment the recently proposed data-driven MPC scheme from~\cite{Berberich19b} by a constraint tightening to guarantee closed-loop output constraint satisfaction despite noisy output measurements.
The key contributions are (i) the constraint tightening, (ii) a proof of recursive feasibility and closed-loop constraint satisfaction, and (iii) data-based estimation procedures for two system constants related to controllability and observability of the unknown system.
Since the proposed MPC scheme uses noisy data to predict future trajectories, our result can also be seen as an alternative to model-based output-feedback MPC in the presence of parametric uncertainty and measurement noise.

\bibliographystyle{IEEEtran}   
\bibliography{Literature}  

\begin{thebibliography}{10}
\providecommand{\url}[1]{#1}
\csname url@rmstyle\endcsname
\providecommand{\newblock}{\relax}
\providecommand{\bibinfo}[2]{#2}
\providecommand\BIBentrySTDinterwordspacing{\spaceskip=0pt\relax}
\providecommand\BIBentryALTinterwordstretchfactor{4}
\providecommand\BIBentryALTinterwordspacing{\spaceskip=\fontdimen2\font plus
\BIBentryALTinterwordstretchfactor\fontdimen3\font minus
  \fontdimen4\font\relax}
\providecommand\BIBforeignlanguage[2]{{%
\expandafter\ifx\csname l@#1\endcsname\relax
\typeout{** WARNING: IEEEtran.bst: No hyphenation pattern has been}%
\typeout{** loaded for the language `#1'. Using the pattern for}%
\typeout{** the default language instead.}%
\else
\language=\csname l@#1\endcsname
\fi
#2}}

\bibitem{Rawlings17}
J.~B. Rawlings, D.~Q. Mayne, and M.~M. Diehl, \emph{Model Predictive Control:
  Theory and Design}, 2nd~ed.\hskip 1em plus 0.5em minus 0.4em\relax Nob Hill
  Pub, 2017.

\bibitem{Yang15}
H.~Yang and S.~Li, ``A data-driven predictive controller design based on
  reduced hankel matrix,'' in \emph{Proceedings of the 10th Asian Control
  Conference}, 2015, pp. 1--7.

\bibitem{Coulson19}
J.~Coulson, J.~Lygeros, and F.~D{\"o}rfler, ``Data-enabled predictive control:
  in the shallows of the {DeePC},'' in \emph{Proceedings of the 18th European
  Control Conference}, 2019, pp. 307--312.

\bibitem{Coulson19b}
------, ``Regularized and distributionally robust data-enabled predictive
  control,'' in \emph{Proc. IEEE Conference on Decision and Control}, 2019, pp.
  2696--2701.

\bibitem{Berberich19b}
J.~Berberich, J.~K{\"o}hler, M.~A. M{\"u}ller, and F.~Allg{\"o}wer,
  ``Data-driven model predictive control with stability and robustness
  guarantees,'' \emph{IEEE Transactions on Automatic Control}, 2021, doi:
  10.1109/TAC.2020.3000182.

\bibitem{Berberich20c}
------, ``Data-driven tracking {MPC} for changing setpoints,'' in \emph{Proc.
  IFAC World Congress}, 2020, to appear, preprint online: arXiv:1910.09443.

\bibitem{Willems05}
J.~C. Willems, P.~Rapisarda, I.~Markovsky, and B.~{De Moor}, ``A note on
  persistency of excitation,'' \emph{Systems \& Control Letters}, vol.~54, pp.
  325--329, 2005.

\bibitem{Chisci02}
L.~Chisci and G.~Zappa, ``Feasibility in predictive control of constrained
  linear systems: the output feedback case,'' \emph{International Journal of
  Robust and Nonlinear Control}, vol.~12, no.~5, pp. 465--487, 2002.

\bibitem{Mayne06}
D.~Q. Mayne, S.~V. Rakovi{\'c}, R.~Findeisen, and F.~Allg{\"o}wer, ``Robust
  output feedback model predictive control of constrained linear systems,''
  \emph{Automatica}, vol.~42, pp. 1217--1222, 2006.

\bibitem{Lorenzen19}
M.~Lorenzen, M.~Cannon, and F.~Allg{\"o}wer, ``Robust {MPC} with recursive
  model update,'' \emph{Automatica}, vol. 103, no. 461-471, 2019.

\bibitem{Aswani13}
A.~Aswani, H.~Gonzalez, S.~S. Sastry, and C.~Tomlin, ``Provably safe and robust
  learning-based model predictive control,'' \emph{Automatica}, vol.~49, no.~5,
  pp. 1216--1226, 2013.

\bibitem{Hewing20}
L.~Hewing, K.~P. Wabersich, M.~Menner, and M.~N. Zeilinger, ``Learning-based
  model predictive control: toward safe learning in control,'' \emph{Annual
  Review of Control, Robotics, and Autonomous Systems}, vol.~3, pp. 269--296,
  2020.

\bibitem{tanaskovic2014adaptive}
M.~Tanaskovic, L.~Fagiano, R.~Smith, and M.~Morari, ``Adaptive receding horizon
  control for constrained mimo systems,'' \emph{Automatica}, vol.~50, no.~12,
  pp. 3019--3029, 2014.

\bibitem{Terzi19}
E.~Terzi, L.~Fagiano, M.~Farina, and R.~Scattolini, ``Learning-based predictive
  control for linear systems: a unitary approach,'' \emph{Automatica}, vol.
  108, p. 108473, 2019.

\bibitem{Berberich20a}
J.~Berberich and F.~Allg\"ower, ``A trajectory-based framework for data-driven
  system analysis and control,'' in \emph{Proc. European Control Conference},
  2020, pp. 1365--1370.

\bibitem{Persis19}
C.~{De Persis} and P.~Tesi, ``Formulas for data-driven control: Stabilization,
  optimality and robustness,'' \emph{IEEE Transactions on Automatic Control},
  vol.~65, no.~3, pp. 909--924, 2020.

\bibitem{Goodwin14}
G.~C. Goodwin and K.~S. Sin, \emph{Adaptive filtering prediction and
  control}.\hskip 1em plus 0.5em minus 0.4em\relax Courier Corporation, 2014.

\bibitem{Markovsky08}
I.~Markovsky and P.~Rapisarda, ``Data-driven simulation and control,''
  \emph{International Journal of Control}, vol.~81, no.~12, pp. 1946--1959,
  2008.

\bibitem{cai2008input}
C.~Cai and A.~R. Teel, ``Input--output-to-state stability for discrete-time
  systems,'' \emph{Automatica}, vol.~44, no.~2, pp. 326--336, 2008.

\bibitem{Koch20}
A.~Koch, J.~Berberich, J.~K{\"o}hler, and F.~Allg{\"o}wer, ``Determining
  optimal input-output properties: A data-driven approach,''
  \emph{arXiv:2002.03882}, 2020.

\bibitem{Herceg13}
M.~Herceg, M.~Kvasnica, C.~Jones, and M.~Morari, ``{Multi-Parametric Toolbox
  3.0},'' in \emph{Proc.~of the European Control Conference}, 2013, pp.
  502--510.

\bibitem{MOSEK15}
{MOSEK ApS}. (2015) The {MOSEK} optimization toolbox for {MATLAB} manual,
  version 7.1 (revision 28).

\end{thebibliography}

\end{document}